\documentclass[conference]{IEEEtran}

\usepackage{soul}
\usepackage{graphicx}
\usepackage{amsmath}
\usepackage{subcaption}
\usepackage{graphicx}
\usepackage{subfloat}
\usepackage{nnfootnote}
\usepackage{amssymb}
\usepackage{epsfig}
\usepackage{stfloats}
\usepackage{ifpdf}
\usepackage{amsmath}
\usepackage{cite}
\usepackage{times}
\usepackage{epsfig}
\usepackage{amsfonts}
\usepackage{ifpdf}
\usepackage{shortvrb,latexsym}
\usepackage[colorinlistoftodos]{todonotes}
\usepackage{hyperref}
\hypersetup{
  colorlinks=true,
  linkcolor=black,
  urlcolor=black,
  citecolor=black,
}

\usepackage{tikz}
\usetikzlibrary{quantikz2}

\usepackage{algorithm}
\usepackage[noend]{algpseudocode}

\usepackage{balance}

\usepackage{glossaries}
\newacronym{hhl}{HHL}{Harrow–Hassidim–Lloyd}
\newacronym{ai}{AI}{Artificial Intelligence}
\newacronym{ml}{ML}{Machine Learning}
\newacronym{pqc}{PQC}{Post-Quantum Cryptography}
\newacronym{qkd}{QKD}{Quantum Key Distribution}
\newacronym{qml}{QML}{Quantum ML}
\newacronym{qpp}{QPP}{Quantum Permutation Pad}
\newacronym{xqml}{XQML}{eXplainable Quantum ML}
\newacronym{fpga}{FPGA}{Field Programmable Gate Arrays}

\usepackage{amsthm}
\newtheorem{theorem}{Theorem}[section]

\newtheorem{lemma}[theorem]{Lemma}

\newtheorem{remark}{Remark}[section]

\usepackage{braket}
\usepackage{mdframed}
\usepackage{lipsum,adjustbox}

\newif\ifcomments

\commentsfalse 

\begin{document}

\title{Quantum CORDIC --- Arcsine on a Budget}

\author{
  \IEEEauthorblockN{Iain Burge}
  \IEEEauthorblockA{SAMOVAR, Institut Polytechnique\\ de Paris, Palaiseau, France \\
    iain-james.burge@telecom-sudparis.eu}
  \and
  \IEEEauthorblockN{Michel Barbeau}
  \IEEEauthorblockA{Carleton University, School\\ of Computer Science, Ottawa, Canada\\
    barbeau@scs.carleton.ca}
  \and
  \IEEEauthorblockN{Joaquin Garcia-Alfaro}
  \IEEEauthorblockA{SAMOVAR, Institut Polytechnique\\ de Paris, Palaiseau, France \\
    joaquin.garcia\_alfaro@telecom-sudparis.eu\\
  }
}

\newcommand{\etal}{\textit{et al}. }
\newcommand{\ie}{\textit{i.e.}, }
\newcommand{\cf}{\textit{cf.}, }
\newcommand{\etc}{\textit{etc.}}
\newcommand{\ex}{\textit{e.g.}, }

\maketitle

\begin{abstract}
This work introduces a quantum algorithm for computing the function arcsine, with arbitrary accuracy. 
We leverage a technique from embedded computing and field-programmable gate arrays, called COordinate Rotation DIgital Computer (CORDIC).
CORDIC is a family of iterative algorithms that, in a classical context, can approximate various trigonometric, hyperbolic, and elementary functions using only bit shifts and additions.
Adapting CORDIC to the quantum context is non-trivial, as the algorithm traditionally uses several non-reversible operations.
We detail a method for CORDIC that avoids such non-reversible operations.
We propose an approach to calculate the arcsine function reversibly with CORDIC.
For n bits of precision, our method has space complexity of order n qubits, a layer count in the order of n times log n, and a \texttt{CNOT} count in the order of n squared. 
This primitive function is a required step for the Harrow–Hassidim–Lloyd (HHL) algorithm, is necessary for quantum digital-to-analog conversion, can simplify a quantum speed-up for Monte-Carlo methods, and has direct applications in the quantum estimation of Shapley values. 

\end{abstract}

\IEEEpeerreviewmaketitle

\section{Introduction}
\label{sec:introduction}

\noindent The problem of quantum digital-to-amplitude conversion \cite{mitarai2019quantum} is an essential step in various quantum algorithms, including \gls{hhl} --- a quantum algorithm for solving linear equations~\cite{harrow2009quantum}, a quantum speed-up for Monte Carlo methods \cite{montanaro2015quantum}, and a quantum algorithm for Shapley value estimation \cite{burge2023quantum}. However, to efficiently perform this conversion it is necessary to calculate inverse trigonometric functions, which is computationally expensive \cite{mitarai2019quantum}.
As quantum computation is in its early stages, it is valuable to consider techniques used in early classical computing intended for weak hardware. 
In particular, this work adapts the classical CORDIC algorithm for \texttt{arcsin}\footnote{Conventionally, the name of the function is acrsine, while the function itself is written \texttt{arcsin}.} to the quantum setting.

Section~\ref{sec:relwork} presents related work. Sections~\ref{sec:preliminaries} and~\ref{sec:cCORDIC} provide  preliminaries to our problem, and introduces the approach of classical CORDIC to approximate \texttt{arcsin} on minimal hardware. Section~\ref{sec:qCORDIC} translates the approach to a quantum setting. Section~\ref{sec:tranImplem} demonstrates the use of quantum CORDIC to perform a quantum digital-to-analogue transformation.
Section~\ref{sec:complexity} gives a sketch proof of the time and space complexity of quantum CORDIC and provides simulation results. 
Section~\ref{sec:conc} concludes the work.

\section{Related Work}
\label{sec:relwork}

\noindent There have been a few proposed algorithms for the calculation of arcsine on quantum computers. 
For instance, H\"{a}ner et al.~\cite{haner2018optimizing} leverages quantum parallelism to perform piecewise polynomial approximations of a wide range of functions.
It can be seen as a promising solution that scales well in precise situations, e.g., more than $32$ bits.
The approach is flexible, but may not be ideal in the short term, since it requires some memory access, a fair number of multiplications, and square rooting for the extreme portions of $\arcsin{x}$, i.e., $x\in [-1,-0.5]\cup[0.5,1]$.
There also exists the quantum function-value binary
expansion method which performs a binary search to approximate $\arcsin{x}$, among other elementary functions~\cite{wang2020quantum}.
The solution is broadly applicable and has some interesting parallels to the CORDIC algorithm. 
However, the arcsine implementation requires a squaring operation per bit of output, which is difficult to implement in the near term.

The above techniques are valuable for various applications when sufficient hardware is available. 
For instances where we do not have those resources available, we consider a technique from early computing.
In the late 1950s, Volder~\cite{volder1959cordic} introduced the CORDIC algorithm to solve an important family of trigonometric equations using the underpowered hardware of the time. 
This was later generalized to a single algorithm that solves various elementary functions, including multiplication, division, the trigonometric functions, \texttt{arctan}, the hyperbolic trigonometric functions, $\ln$, $\exp$, and square root \cite{walther1971unified}. 
A positive trait is that many of these operations can exist with the same architecture by leveraging simple parameters \cite{lakshmi2010cordic}.
CORDIC has been further modified to approximate various other functions, including our target function, \texttt{arcsin}~\cite{mazenc1993computing}.

\section{Preliminaries}
\label{sec:preliminaries}

\noindent Suppose you have the following $(n+1)$-qubit quantum state:
\begin{equation*}
    \ket{\phi} = \sum\limits_{k=0}^{N-1} \alpha_k \ket{h_k}_\texttt{in} \ket{0}_\texttt{out},
\end{equation*}
where $N = 2^{n}$, the \texttt{in} register has size $n$, the \texttt{out} register is of size one, and $h_k$ is the fixed point number $k2^{-n}$.
A particularly useful transformation on this state is given by the following equation:
\begin{equation}\label{eq:goalTransform}
    U\ket{\phi} = \sum\limits_{k=0}^{N-1} \alpha_k \ket{h_k}_\texttt{in} \left(\sqrt{1-h_k} \ket{0} + \sqrt{h_k} \ket{1}\right)_\texttt{out},
\end{equation}
which encodes the binary values $h_k$ into the probability amplitudes of the output register. This transformation is an important instance of quantum digital-to-analog conversion~\cite{mitarai2019quantum}.

A fast solution to this problem is valuable in multiple algorithms.
For instance, a quantum algorithm for approximating Shapley values \cite{burge2023quantum} could leverage this transformation to simplify the state preparation step.
Another important example is the quantum speedup of Monte Carlo methods \cite{montanaro2015quantum} which requires the transformation $W$, defined as:
\begin{equation*}
    \ket{x}_\texttt{in}\ket{0}_\texttt{out}
    \xrightarrow{W}
    \ket{x}_\texttt{in}\left(\sqrt{1-\Phi(x)}\ket{0} + \sqrt{\Phi(x)}\ket{1} \right)_\texttt{out}.
\end{equation*}
where $\Phi$ is a function from binary strings to the interval from zero to one.
This can be translated easily to our transformation.

We split the computation into two steps, where we first implement a reversible algorithm to compute $\Phi$ in the computational basis, and then apply our transformation.
Take register \texttt{in} to have $m$ qubits, the register \texttt{aux} to have $n$ qubits, and the register \texttt{out} to have one qubit. We first apply the unitary $V$ using the \texttt{in} register as input and outputting the $n$-bit approximate result $\tilde{\Phi}$ to \texttt{aux}, as follows:
\begin{equation*}
    \ket{x}_\texttt{in}\ket{0}_\texttt{aux} \ket0_{\texttt{out}} \xrightarrow{V} \ket{x}_\texttt{in}\ket{\tilde{\Phi}(x)}_\texttt{aux}\ket{0}_\texttt{out}.
\end{equation*}
where $\ket{x}_\texttt{in}$ is a computational basis vector, and $\tilde\Phi(x)$ is an $n$ bit approximation of $\Phi(x)$.
Next we perform $U$, Equation~\eqref{eq:goalTransform}, on the resulting \texttt{aux} and \texttt{out} register giving,
\begin{equation*}
    \ket{x}_\texttt{in} \ket{\tilde\Phi(x)}_\texttt{aux}\left(\sqrt{1-\tilde\Phi(x)}\ket{0} + \sqrt{\tilde\Phi(x)}\ket{1} \right)_\texttt{out}.
\end{equation*}
Finally, to restore the \texttt{aux} register to its initial state, we apply the $V^{-1}$ transformation.
In summary, we can solve the problem using a reversible classical algorithm implementation of $\tilde{\Phi}$. 
We complete the process by encoding the result into the probability amplitude of our output bit.

Another practical use case is as a step of the \gls*{hhl} algorithm~\cite{harrow2009quantum}.
In this case, reciprocal eigenvalues are encoded as binary strings in a superposition.
An essential step is to encode the reciprocals of the eigenvalues in the probability amplitudes of an output bit.
In particular, we can apply a transformation closely related to $U$ from Equation~\eqref{eq:goalTransform}.
Hence, an efficient algorithm to implement the transformation $U$, has immediate benefits to foundational problems.

There are a few naive approaches to implementing the transformation U, Equation~\eqref{eq:goalTransform}.
For example, one to use a lookup table, where each possible input $\ket{h_k}_\texttt{in}$ is separately implemented, but this would require exponential circuit depth.
qRAM could be leveraged \cite{giovannetti2008quantum}, however it would require exponential space.
These approaches could be combined with linear interpolation to perform well at very low precision, but they do not scale well for our suggested applications.

With some naive methods out of the way, we examine a less trivial direction that illustrates the main challenge solved by this paper. Consider the following state,
\begin{equation*}
    \ket{\psi_0} = \sum\limits^{N-1}_{k=0} \alpha_k \ket{h_k}_\texttt{in} \ket{0}_\texttt{aux} \ket{0}_\texttt{out}.
\end{equation*}
where \texttt{in} and \texttt{aux} are $n$-qubit registers, each representing two's complement, fixed point integers in the range $[-2,2)$ of the form $x_0x_1.x_2\cdots x_{n-1}$ where $x_0$ is the sign bit, and \texttt{out} is a single qubit register.
We require $h_k$ to be in range $[0,1)$, meaning if $k>N/4$ then $\alpha_k=0$.
Suppose we apply the function $\arcsin\sqrt{\cdot}$ where \texttt{in} is the input register, and \texttt{aux} is the output register.
Applying this unitary, which we call $F$, yields,
\begin{equation*}
    F\ket{\psi_0} = \ket{\psi_1} = \sum\limits^{N-1}_{k=0} \alpha_k \ket{h_k}_\texttt{in} \ket{\arcsin\sqrt{h_k}}_\texttt{aux} \ket{0}_\texttt{out}.
\end{equation*}

Suppose now that we apply the quantum circuit $R$ from Figure~\ref{fig:encodeBinaryAsAmplitude}.
This would give us state,
\begin{align*}
    R&\ket{\psi_1} = \ket{\psi_2} = \sum\limits^{N-1}_{k=0} \alpha_k \ket{h_k}_\texttt{in} \ket{\arcsin\sqrt{h_k}}_\texttt{aux} \\
    &\cdot \left(\cos\left(\arcsin \sqrt{h_k} \right)\ket{0} + \sin\left(\arcsin \sqrt{h_k} \right)\ket{1}\right)_\texttt{out}.
\end{align*}
Using trigonometry identities, we have that,
\begin{scriptsize}
\begin{align*}
    \ket{\psi_2} 
    = \sum\limits^{N-1}_{k=0} \alpha_k \ket{h_k}_\texttt{in} \ket{\arcsin\sqrt{h_k}}_\texttt{aux} \cdot \left(\sqrt{1-h_k} \ket{0} + \sqrt{h_k} \ket{1}\right)_\texttt{out}.
\end{align*}
\end{scriptsize}
Uncomputing $\arcsin\sqrt\cdot$ by performing $F^{-1}$, we get our desired state, Equation~\eqref{eq:goalTransform}.
This approach transforms our previous challenge into the problem of computing $\arcsin$ reversibly.
This work describes a method to achieve this task efficiently.

\begin{figure}[!t]
\begin{center}
    \begin{adjustbox}{width=\columnwidth}
    \begin{tikzpicture}
    \node[scale=1] {
    \begin{quantikz}
    \lstick[wires=4]{$\ket{\arcsin\sqrt{h_k}}_\texttt{aux}$} &
    \ctrl{4} & \qw & \qw & \qw & \qw &\qw &\qw &
    \\
    &  \qw & \ctrl{3} & \qw & \qw & \qw &\qw &\qw &
    \\
    &\setwiretype{n}\vdots &  & & \ddots & &
    \\
    & \qw & \qw & \qw & \qw & \qw & &\ctrl{1} &\qw 
    \\
    \lstick{$\ket{0}_\texttt{out}$} & \gate{R_{y}\left(\frac{\pi}{4}\right)} & \gate{R_{y}\left(\frac{\pi}{8}\right)} & \qw & \setwiretype{n}\cdots & & \setwiretype{q} & \gate{R_{y}\left(\frac{\pi}{2^{n+1}}\right)} & 
    \end{quantikz}
    };
    \end{tikzpicture}
    \end{adjustbox}
\end{center}
\caption{Unitary transformation $R$ which encodes the binary value of the \texttt{aux} register into the output register. 
Define \hbox{$R_{y}(\omega)=(\cos\omega, -\sin\omega; \sin\omega, \cos\omega)$}.
}
\label{fig:encodeBinaryAsAmplitude}
\end{figure}
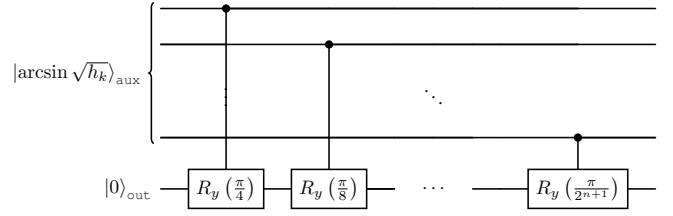

\section{Classical CORDIC for Arcsine}
\label{sec:cCORDIC}

\begin{algorithm}
  \caption{CORDIC Arcsine}\label{al:arcsin}
  \begin{algorithmic}[1]
    \Procedure{Arcsin}{$t$, $\texttt{\#iter}$}\Comment{$t \in [-1,1]$}\Comment{$s(z)=1$ if $z<0$ else $s(z)=0$}
      \State $\texttt{ang} \gets 0$, $\texttt{x} \gets 1$, $\texttt{y} \gets 0$, $\texttt{t} \gets t, \texttt{d} \gets 0$ \label{line:asin-init}
      \For{$i=1, i<\texttt{\#iter}, \texttt{++}i$}
        \State{\footnotesize $\texttt{d}_i \gets [s(\texttt{x}) \land s(\texttt{t}-\texttt{y})] \oplus s(\texttt{x}) \oplus [s(\texttt{x})\land s(\texttt{y})]\oplus s(\texttt{t}-\texttt{y})$\label{line:asin-dUpdate}}
        \If{$\texttt{d}_i$}
          $\texttt{x}, \texttt{y} \gets \texttt{y}, \texttt{x}$ 
        \EndIf
        \For{\texttt{\_ in range(2)}} \label{line:pseudo-roation}
            \State $\texttt{x}, \texttt{y} \gets \texttt{x}-2^{-i}\texttt{y}, \texttt{y}+2^{-i}\texttt{x}$
        \EndFor
        \If{$\texttt{d}_i$}
          $\texttt{x}, \texttt{y} \gets \texttt{y}, \texttt{x}$ 
        \EndIf
        \State $\texttt{t} \gets (1+2^{-2i})\texttt{t}$
        \State $\texttt{ang} \gets \texttt{ang} + (-1)^{\texttt{d}_i}2\arctan{2^{-i}}$ \label{line:asin-thetaUpdate}
      \EndFor
      \State \textbf{return} $\texttt{ang}$
    \EndProcedure
  \end{algorithmic}
\end{algorithm}
\noindent Let us give a brief insight into how CORDIC works. 
To begin, there is an implicit two-dimensional goal vector that encodes the problem. 
The exact vector is unknown, but we have a method to compare it to other vectors. 
A starting vector rotates towards the goal in increasingly small discrete steps, picking the direction of rotation based on our comparison method, until it converges with satisfactory error. 
It is possible to avoid performing full multiplications during the rotations through the use of pseudo-rotations, rotations with a small stretch.
See Figure~\ref{fig:ExCORDIC} for an explicit example.

Our proposed method is based on Mazenc et al.'s algorithm~\cite{mazenc1993computing}.
In this work, Step~1 is slightly improved to better handle an edge case. 
The procedure is fully described in Algorithm~1. 
Each iteration follows the steps below:

\begin{figure}[!hptb]
\centering
\includegraphics[width=.48\textwidth]{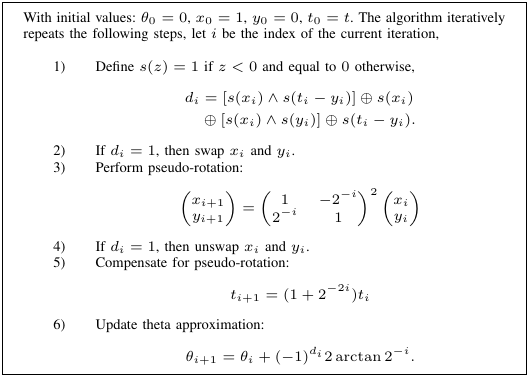}
\end{figure}

We first initialize the problem, given input $t\in[-1,1]$.
Our goal is to get the vector $(x_i,y_i)$ to point in the same direction as the vector of length $1$ with height $t$, i.e., vector $(\sqrt{1-t^2},t)$. 
This goal vector has an angle of $\theta=\arcsin(t)$.
By keeping track of each rotation we make, assuming $(x_i,y_i)$ points in roughly the same direction as the goal vector, we can construct an approximation for $\theta$.
In particular, assuming infinite precision, as $i$ tends to infinity, approximation $\theta_i$ tends to $\theta$.
According to Ref.~\cite{mazenc1993computing}, to achieve $n$ bits of accuracy it takes $n+2$ correct iterations.

Step~1 is a computationally convenient way to check if $\theta_i < \theta$, which we store in variable $d_i$.
If $d_i=0$, then we need to rotate counterclockwise to get closer to $\theta$.
Otherwise, if $d_i=1$, we need to rotate clockwise.
Steps~2 and~4 conjugate Step~3 to achieve the correct rotation direction.
Step~3 performs a pseudo-rotation, slightly stretching our vector $(x_i, y_i)$ by a factor of $1+2^{-2i}$.
It is performed by making the following substitutions twice: $x'=x-2^{-i}y$ and $y'=y+2^{-i}x$.
As a result of the operation, we point closer to the goal direction, but it is necessary to compensate for the stretch.
Step~5 deals with the stretching of the vector $(x_i,y_i)$ by stretching our goal vector by the same amount, $(1+2^{-2i})$.
Note that multiplication by powers of $2$ can be done efficiently through bit-shifting.
Finally, Step~6 records the change to the current angle of vector $(x_i,y_i)$.
Note that the expression $2\arctan2^{-i}$ can be precomputed and directly encoded into hardware.

\begin{algorithm}
  \caption{Mult($\texttt{in},\texttt{aux},m$): $\texttt{in}\gets (1+2^{-m})\texttt{in}$}\label{al:mult}
  \begin{algorithmic}[1]
    \Procedure{Mult}{$\texttt{in}$, \texttt{aux}, $m$}
      \Comment{$\texttt{in}$ is a register of size $n$, \texttt{aux} is a auxiliary register of size $n$ with near 0 value, $m$ is the right shift}
      \State{\texttt{F} $\gets [1,1,2,3,5,8,13,\dots]$}  \Comment{Fibbonacci Sequence}
      \State{$\texttt{\#iter} \gets 2\left\lceil\log_\varphi(\sqrt{5}n/m)/2\right\rceil$} \Comment{$\varphi = (1+\sqrt{5})/2$}
      \State{$\texttt{aux} \gets \texttt{aux}+\texttt{in}$} \label{line:initStart}
      \For{$i=\texttt{\#iter}$, $i\geq 0$, $\texttt{--}i$} \label{line:unbuildX}
        \If{$i$ even}
          \State $\texttt{in} \gets \texttt{in}-(-1)^{\texttt{F}[i]}(\texttt{aux}\gg(m\cdot \texttt{F}[i]))$ 
        \Else
          \State $\texttt{aux} \gets \texttt{aux}-(-1)^{\texttt{F}[i]}(\texttt{in}\gg(m\cdot \texttt{F}[i]))$
        \EndIf
      \EndFor
      \State $\texttt{aux} \gets \texttt{aux}-\texttt{in}$\label{line:emptyAux2}
    \EndProcedure
  \end{algorithmic}
\end{algorithm}

\begin{figure*}
     \centering
     \begin{subfigure}[b]{0.32\textwidth}
         \centering
         \includegraphics[width=\textwidth]{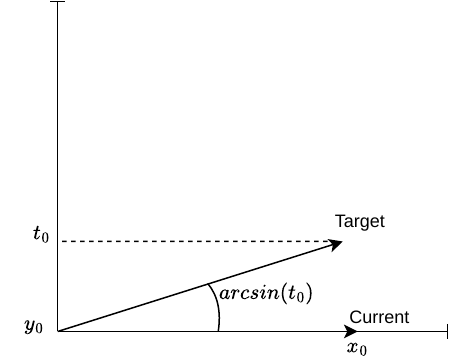}
         \caption{Initial values, target vector is \hbox{$((1-t_0^2)^{1/2},t_0)\approx(.96,.29)$}, note that the first component is implicit, and we need not calculate it.
         The current vector is $(x_0,y_0)=(1,0)$.
         Performing Line~\ref{line:asin-dUpdate} shows that our next pseudo-rotation is counterclockwise.}
         \label{fig:ExCORDICiter0}
     \end{subfigure}
     \hfill
     \begin{subfigure}[b]{0.32\textwidth}
         \centering
         \includegraphics[width=\textwidth]{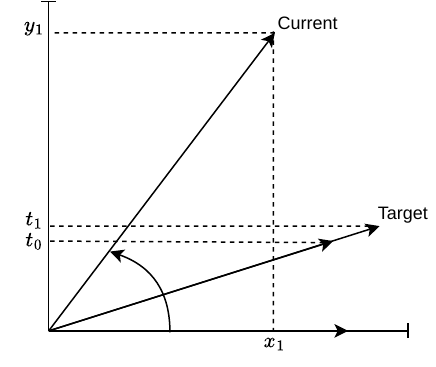}
         \caption{
         The \texttt{for} loop beginning at Line~\ref{line:pseudo-roation} performs the pseudo-rotation on our Current vector, yielding $(x_1,y_1)=(0.75,1)$.
         The target vector is stretched by a factor of $5/4$ to compensate for the pseudo-rotation, giving us $t_1\approx0.36$.
         Line~\ref{line:asin-dUpdate} determines that the next pseudo-rotation is clockwise.
         }
         \label{fig:ExCORDICiter1}
     \end{subfigure}
     \hfill
     \begin{subfigure}[b]{0.32\textwidth}
         \centering
         \includegraphics[width=\textwidth]{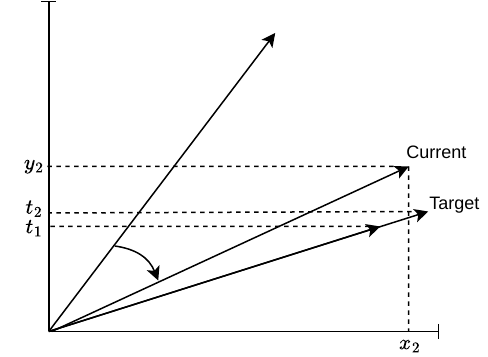}
         \caption{
         The \texttt{for} loop beginning at Line~\ref{line:pseudo-roation} performs the pseudo-rotation on our Current vector, yielding $(x_2,y_2)\approx(1.21,0.56)$.
         The target vector is stretched further by a factor of $17/16$ to compensate for the pseudo-rotation, giving $t_2\approx0.39$.
         Line~\ref{line:asin-dUpdate} determines that the next pseudo-rotation is clockwise. }
         \label{fig:ExCORDICiter2}
     \end{subfigure}
        \caption{
        Example of first two iterations of Algorithm~\ref{al:arcsin} Arcsine for $t\approx0.29$.
        The target vector has an angle of $\arcsin(t)$, note that stretching the vector does not impact the angle of the vector.
        Each pseudo-rotation is recorded and the total change in angle is summed, yielding an approximation for $\arcsin(t)$.
        These results are generated when $t$ is $300/1024$ represented with $12$ bits as a fixed point integer in $[-2,2)$ and can be viewed in the companion GitHub repository~\cite{Burge_2024}.
        }
        \label{fig:ExCORDIC}
\end{figure*}

\section{Quantum CORDIC for Arcsine}
\label{sec:qCORDIC}
There are a few immediate challenges that appear when adapting CORDIC to a quantum setting, in this section we detail the solutions step-by-step.
We first note a few important facts.
We represent the values $\theta_i, x_i, y_i, t_i$ using fixed point two's complement in the range $[-2,2)$ with one sign bit, one integer bit, and \hbox{$n-2$} fractional bits.
This avoids any issues with overflow, in the case of $\theta_i$, $\theta$ is in the range $[-\pi/2, \pi/2)$, so $|\theta| \leq \pi/2 < 2$.
In the case of $|x_i|, |y_i|, |t_i|$, each represents side-lengths of vectors with a magnitude that increases in size by $(1+4^{-i})$ each step, thus,
\begin{equation*}
    |x_i|, |y_i|, |t_i| \leq \prod\limits_{i=1}^{n} \left(1+4^{-i} \right) <  \sqrt[3]{e} < 2.
\end{equation*}
One can verify the inequality by taking $n\rightarrow \infty$, taking the $\ln$ of the infinite product, using the inequality \hbox{$\ln(1+x)<x$} for positive $x$, and finally using the geometric series.
The product represents the stretch caused by the pseudo-rotations over the iterations.
The two's complement also allows us to check if a value is negative easily using the most significant bit, if it is $1$ then the value is negative, otherwise, the value is positive.
To approximate multiplication by $2^{-m}$, we use the right bit-shift operation, $x_0x_1\cdots x_{n-1}\gg m$, we move each bit to the right by $m$ places, ignoring the $m$ rightmost bits.  
Also note that, when right bit-shifting in two's complement, the previous leftmost bit is copied into the new leftmost bits, for example:
\begin{equation*}
    x_0 x_1 x_2 x_3 x_4 \gg 2 = x_0 x_0 x_0 x_1 x_2 , \quad x_k\in\{0,1\}.
\end{equation*}
This ensures right bitshifts work as a division by powers of two for negative numbers. %

Now, let us go through the steps of the quantum implementation, assuming $n$ bits each to represent $\theta_i, x_i, y_i, t_i$ as fixed point numbers in registers $\texttt{ang}, \texttt{x}, \texttt{y}, \texttt{t}$ respectively.
Allocate $n-1$ bits to represent \hbox{$d=d_1d_2\cdots d_{n-1}$} as a bit array in register \texttt{d} where $\texttt{d}_i$ represents the $i$th register qubit.
Initialization is trivial, we use $\texttt{t}\gets t$ as the input. 
Note that we do \emph{not} need to use new registers for each iteration, one register per value is sufficient.
Additionally, in this work we consider the leftmost bit is most significant.

We perform Step~1 with the following sequence of operations. 
We begin with the state:
\begin{equation*}
    \ket{0}_{\texttt{d}_i} \ket{x_i}_\texttt{x} \ket{y_i}_\texttt{y} \ket{t_i}_\texttt{t}
\end{equation*}
Where $\texttt{d}_i$ is the $i$th bit of the \texttt{d} register.
We perform a subtraction (implemented with an inverse addition gate) from \texttt{y} to \texttt{t}, yielding,
\begin{equation*}
    \ket{0}_{\texttt{d}_i} \ket{x_i}_\texttt{x} \ket{y_i}_\texttt{y} \ket{t_i-y_i}_\texttt{t}.
\end{equation*}
Looking directly at the sign bits of the \texttt{x}, \texttt{y}, and \texttt{t} registers gives $s(x_i)$, $s(y_i)$, and $s(t_i-y_i)$ respectively, where \hbox{$s(z)=1$} if $z<0$ else $s(z)=0$.
We next perform the following operations,
\begin{itemize}
    \item Perform a \texttt{Toffoli} gate using the sign bits of the \texttt{x} and \texttt{t} registers as controls and qubit $\texttt{d}_i$ as the target.
    \item Perform a \texttt{Toffoli} gate using the sign bits of the \texttt{x} and \texttt{y} registers as controls and qubit $\texttt{d}_i$ as the target.
    \item Perform a \texttt{CNOT} gate using the sign bits of the \texttt{x} register as control and qubit $\texttt{d}_i$ as the target.
    \item Perform a \texttt{CNOT} gate using the sign bits of the \texttt{t} register as control and qubit $\texttt{d}_i$ as the target.
\end{itemize}
Finally, we can undo the previous subtraction by adding the \texttt{y} register back to the \texttt{t} register.
This gives us a correct $d_i$ bit without side effects.

Steps~2 and~4 are trivial to implement using swap gates between each bit of the \texttt{x} and \texttt{y} register controlled by $\texttt{d}_i$, this can be done using $18n$ \texttt{CNOT}s.
Step~3 is the most difficult.
It leverages Algorithm~2 for multiplying by $(1+2^{-k})$.
Note, that each step of Algorithm~2 is reversible.
An explanation for \texttt{Mult} is the subject of Appendix~\ref{appendix:mult}.
For notational simplicity, we assume \texttt{x} and \texttt{y} have not been swapped.
Our initial state for this step is,
\begin{equation*}
    \ket{x_i}_\texttt{x}\ket{y_i}_\texttt{y}.
\end{equation*}
We first subtract $y_i \gg i$ ($y_i$ right shifted by $i$ bits) from register \texttt{x}, giving us approximately, 
\begin{equation*}
    \ket{x_i-2^{-i}y_i}_\texttt{x}\ket{y_i}_\texttt{y}.
\end{equation*}
Next, we need to take an auxiliary register with value near $0$.
We apply $\texttt{Mult}(\texttt{y}, \texttt{aux}, 2i)$ to the \texttt{y} register,
\begin{equation*}
    \ket{x_i-2^{-i}y_i}_\texttt{x}\ket{y_i+2^{-2i}y_i}_\texttt{y}.
\end{equation*}
Finally, we add the \texttt{x} register shifted, $(x_i-2^{-i}y_i) \gg i$, to the \texttt{y} register,
\begin{equation*}
    \ket{x_i-2^{-i}y_i}_\texttt{x}\ket{y_i+2^{-2i}y_i+2^{-i}(x_i-2^{-i}y_i)}_\texttt{y}.
\end{equation*}
Which equals $\ket{x_i-2^{-i}y_i}_\texttt{x}\ket{y_i+2^{-i}x_i}_\texttt{y}$.
By repeating these operations once more, Step~4 is accomplished.
Note that the bit shifting and the application of \texttt{MULT} apply a small but compounding error as a result of rounding.
The analysis of this error will be the subject of future work.

With our new $\texttt{Mult}$ tool, Step~5 is trivial.
We simply apply $\texttt{Mult}(\texttt{t}, \texttt{aux}, 2i)$ to the \texttt{t} register.
Finally, for Step~$6$ we apply controlled negation of the $\texttt{ang}$ register using $\texttt{d}_i$ as the control. 
Then add the precomputed $2\arctan2^{-i}$ to the $\theta$ register; this operation can be encoded directly on quantum hardware.
Then, once again, perform a controlled negation of the $\texttt{ang}$ register using $\texttt{d}_i$ as the control.
This sequence gives,
\begin{align*}
    \ket{\theta_i}_\texttt{ang} & \rightarrow \ket{(-1)^{d_i}\theta_i}_\texttt{ang} \rightarrow \ket{(-1)^{d_i}\theta_i + 2\arctan2^{-i}}_\texttt{ang}\\
    &\rightarrow \ket{(-1)^{2d_i}\theta_i + (-1)^{d_i}2\arctan2^{-i}}_\texttt{ang}\\
    &= \ket{\theta_i + (-1)^{d_i}2\arctan2^{-i}}_\texttt{ang} = \ket{\theta_{i+1}}_\texttt{ang}.
\end{align*}
This step can be improved slightly by replacing the negations to the \texttt{ang} register with bitwise \texttt{NOT}s.

Thus, we have a quantum-compatible method for applying each step of the classical CORDIC algorithms.
As we see in Section~\ref{sec:complexity}, these adapted techniques add minimal error in simulation.

\section{Digital to Amplitude Transformation}
\label{sec:tranImplem}
\noindent Suppose we are given a $5n$-qubit quantum state of the form,
\begin{equation*}
    \ket{\kappa_0} = \sum\limits_{k=0}^{N-1} 
    \alpha_k \ket{h_k}_\texttt{t}
    \ket{0}^{\otimes n-1}_\texttt{d}
    \ket{0}^{\otimes n}_\texttt{x}
    \ket{0}^{\otimes n}_\texttt{y}
    \ket{0}^{\otimes n}_\texttt{mult}
    \ket{0}_\texttt{out},
\end{equation*}
where $N=2^{n}$, the \texttt{t} register has size $n$.
Additionally, we have $h_k=k2^{-n}$, where if $k>N/4$ then $\alpha_k=0$.
In other words, register \texttt{t} stores a fixed point binary number between $-2$ and $2$ in two's complement, but we restrict the initial value to be between $0$ and $1$.
The registers \texttt{x, y, mult} are of size $n$, register \texttt{d} is of size $n-1$, and each is initialized to $\ket{0}$.
Finally, the \texttt{out} register is one qubit.
It stores a value corresponding to $h_k$ in its probability amplitudes. 

Using a few modifications to Algorithm~\ref{al:arcsin}, we can leverage CORDIC Arcsine for Digital to Amplitude conversion.
To make use of Remark~\ref{remark:asinOFsqrt}, we first multiply register \texttt{t} by two and subtract one.
Next, we set the \texttt{x} register equal to $1$.
This gives us state,
\begin{equation*}
    \ket{\kappa_1} = \sum\limits_{k=0}^{N-1} \alpha_k \ket{h_k'}_\texttt{t} \ket{0}^{\otimes n}_\texttt{d}\ket{1}_\texttt{x}\ket{0}^{\otimes n}_\texttt{y} \ket{0}^{\otimes n}_\texttt{mult} \ket{0}_\texttt{out},
\end{equation*}
where $h_k'=2h_k-1 \in [-1, 1)$.

We then apply Algorithm~3.
Note that, for our current problem, we do not require an explicit representation of theta.
As a result, we need not include the \texttt{ang} register.
This yields the state,
\begin{equation}\label{eq:kappa2}
    \ket{\kappa_2} = \sum\limits_{k=0}^{N-1} \alpha_k \ket{h_k'}_\texttt{t} \ket{\texttt{DA}(h'_k)}_\texttt{d}\ket{x'}_\texttt{x}\ket{y'}_\texttt{y} \ket{g}_\texttt{mult} \ket{0}_\texttt{out}.
\end{equation}
The $j$th bit of $\texttt{DA}{(h'_k)}$ represents the rotation direction at iteration $j$ of the algorithm.
$x'$, $y'$, and $g$ are the resulting garbage.

\begin{remark} \label{rem:discreteBasi}
    By Mazenc \cite{mazenc1993computing}, we have the following relation,
    \begin{equation*}
        \frac{\arcsin\left(h_k'\right)}{2} \approx \sum\limits_{j=0}^{n}  (-1)^{\texttt{DA}\left(h_k'\right)_j} \arctan \left(2^{-j}\right),
    \end{equation*}
    with error of order $\mathcal{O}(2^{-n})$.
\end{remark}
\begin{remark} \label{remark:asinOFsqrt}
    $\arcsin(\sqrt{x})$ is an affine transformation of $\arcsin (x)$,
    \begin{equation*}
        \arcsin\left(\sqrt{x}\right) = \frac{\arcsin(2x-1)}{2} + \frac{\pi}{4}.
    \end{equation*}
\end{remark}

\noindent Consider, the state $R'\ket{\texttt{DA}(h_k')}_\texttt{d}\ket{0}_\texttt{out}$, which is equal to,
\begin{scriptsize}\begin{align*}
    & \ket{\texttt{DA}(h_k')}_\texttt{d} 
    \Bigg[ \cos\left( \frac{\pi}{4} + \sum_{j=0}^{n}  (-1)^{\texttt{DA}(h_k')_j} \arctan \left(2^{-j}\right) \right) \ket{0} \\
    & + \sin\left( \frac{\pi}{4} + \sum_{j=0}^{n}  (-1)^{\texttt{DA}(h_k')_j} \arctan \left(2^{-j}\right) \right)\ket{1} \Bigg]_\texttt{out}.
\end{align*} \end{scriptsize}
Where $R'$ is implemented in Figure~\ref{fig:encodeDAsAmplitude}. 
Applying $R'$ to the \texttt{d} and \texttt{out} registers of state $\ket{\kappa_2}$, by Remark~\ref{rem:discreteBasi}, gives the approximate state,
\begin{align*}
    \ket{\kappa_3} \approx \sum\limits_{k=0}^{N-1} & \alpha_k \ket{h_k'}_\texttt{t} \ket{\texttt{DA}(h_k')}_\texttt{d}\ket{0}_\texttt{x}^{\otimes n} \ket{0}_\texttt{y} ^{\otimes n} \ket{0}_\texttt{mult}^{\otimes n} \\
    & \otimes \Bigg[ \cos\left( \frac{\pi}{4} + \frac{\arcsin\left(2h_k-1\right)}{2} \right) \ket{0} \\
    & + \sin\left( \frac{\pi}{4} + \frac{\arcsin\left(2h_k-1\right)}{2} \right)\ket{1} \Bigg]_\texttt{out}.
\end{align*}
By Remark~\ref{remark:asinOFsqrt}, this is equal to state,
\begin{scriptsize}
\begin{equation*}
    \sum\limits_{k=0}^{N-1} \alpha_k \ket{h_k'}_\texttt{t} \ket{\texttt{DA}(h_k')}_\texttt{d}\ket{x'}_\texttt{x}^{\otimes n} \ket{y'}_\texttt{y}^{\otimes n} \ket{g}_\texttt{mult}^{\otimes n}
    \left[  \sqrt{1-h_k} \ket{0} + \sqrt{h_k} \ket{1} \right]_\texttt{out}.
\end{equation*}
\end{scriptsize}

The registers \texttt{d}, \texttt{x}, \texttt{y}, \texttt{mult} can be cleaned by performing Algorithm~\ref{al:digToAmp} in reverse. 
We subtract one from register \texttt{x} to zero it.
After adding one to and dividing register \texttt{t} by two and subtracting one from \texttt{x}, the binary to amplitude transformation is complete, yielding,
\begin{scriptsize}
\begin{equation*}
    \ket{\kappa_4} \approx \sum\limits_{k=0}^{N-1} \alpha_k \ket{h_k}_\texttt{t} \ket{0}_\texttt{d}\ket{0}_\texttt{x}^{\otimes n} \ket{0}_\texttt{y}^{\otimes n} \ket{0}_\texttt{mult}^{\otimes n}
    \left[  \sqrt{1-h_k} \ket{0} + \sqrt{h_k} \ket{1} \right]_\texttt{out}.
\end{equation*}
\end{scriptsize}

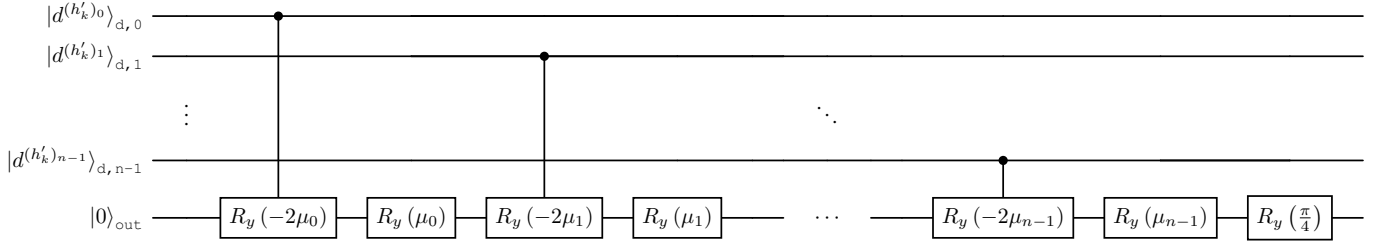
\begin{figure*}
\centering
\begin{adjustbox}{width=\textwidth}
\begin{tikzpicture}
\node[scale=1] {
    \begin{quantikz}
    \lstick{$\ket{\texttt{DA}(h_k')_0}_{\texttt{d}_0}$} &&
    \ctrl{4} &  & \qw & \qw & \qw &\qw && & & & &&
    \\
    \lstick{$\ket{\texttt{DA}(h_k')_1}_{\texttt{d}_1}$} &&
    &  & \ctrl{3} \qw & \qw & \qw &\qw && & & &&&
    \\
    \setwiretype{n} & \vdots &  &&& & && \ddots & &
    \\
    \lstick{$\ket{\texttt{DA}(h_k')_{n-1}}_{\texttt{d}_{n-1}}$} & &&& & && & &&&\ctrl{1} &&\qw &
    \\
    \lstick{$\ket{0}_\texttt{out}$} && \gate{R_y\left(-2\mu_0\right)} & \gate{R_y\left(\mu_0\right)} & \gate{R_y\left(-2\mu_1\right)} & \gate{R_y\left(\mu_1\right)} & && \setwiretype{n}\cdots & & \setwiretype{q} &  \gate{R_y\left(-2\mu_{n-1}\right)} & \gate{R_y\left(\mu_{n-1}\right)} & \gate{R_y\left(\frac{\pi}{4}\right)} &
\end{quantikz}
};
\end{tikzpicture}
\end{adjustbox}
\caption{Unitary transformation $R'$ which transfers the binary encoding of CORDIC rotation direction stored \texttt{d} register into the amplitude of the \texttt{out} register. 
Where $\texttt{DA}(h_k')_i$ represents the $i$th bit of $\texttt{DA}(h_k')$.
Define \hbox{$R_{y}(\omega)=(\cos\omega, -\sin\omega; \sin\omega, \cos\omega)$}, and \hbox{$\mu_i=\tan^{-1}2^{-i}$}.
}
\label{fig:encodeDAsAmplitude}
\end{figure*}

\begin{algorithm}
  \caption{Digital-to-Amplitude: Classical Step}\label{al:digToAmp}
  \begin{algorithmic}[1]
    \Procedure{DA}{$t$, $\texttt{\#iter}$}\Comment{$t \in [-1,1]$}\Comment{$s(z)=1$ if $z<0$ else $s(z)=0$}
      \State $\texttt{x} \gets 1$, $\texttt{y} \gets 0$, $\texttt{t} \gets t, \texttt{d} \gets 0$ \label{line:DAasin-init}
      \For{$i=1, i<\texttt{\#iter}, \texttt{++}i$}
        \State{\footnotesize $\texttt{d}_i \gets [s(\texttt{x}) \land s(\texttt{t}-\texttt{y})] \oplus s(\texttt{x}) \oplus [s(\texttt{x})\land s(\texttt{y})]\oplus s(\texttt{t}-\texttt{y})$\label{line:DAasin-dUpdate}}
        \If{$\texttt{d}_i$}
          $\texttt{x}, \texttt{y} \gets \texttt{y}, \texttt{x}$ 
        \EndIf
        \For{\texttt{\_ in range(2)}}
            \State $\texttt{x}, \texttt{y} \gets \texttt{x}-2^{-i}\texttt{y}, \texttt{y}+2^{-i}\texttt{x}$
        \EndFor
        \If{$\texttt{d}_i$}
          $\texttt{x}, \texttt{y} \gets \texttt{y}, \texttt{x}$ 
        \EndIf
        \State $\texttt{t} \gets (1+2^{-2i})\texttt{t}$
      \EndFor
      \State \textbf{return} $\texttt{d}$
    \EndProcedure
  \end{algorithmic}
\end{algorithm}

\section{Complexity and Simulations}
\label{sec:complexity}

In this section, we detail the complexity to construct $\ket{\kappa_2}$, Equation~\eqref{eq:kappa2}.
The quantum implementation of Algorithm~3 applies \texttt{Mult} three times per iteration, and addition six times per iteration.
\texttt{Mult} takes less than $\sqrt{5}\varphi^{-2i}n + 2$ additions on the $i$th step if $i<n$, where $\varphi$ is the golden ratio, if $i\geq n$, then \texttt{Mult} can be skipped.
Thus, assuming $n$ iterations for $\mathcal{O}(n)$ bits of accuracy, we perform less than,
\begin{align*}
    3\sum\limits_{i=1}^{n/2} \left(\frac{\sqrt{5}n}{\varphi^{2i}} + 2\right) + \sum\limits_{i=1}^{n}6
    &= 9 n + 3\sqrt{5}n \sum\limits_{i=1}^{n/2} \frac{1}{\varphi^{2i}} < 14n,
\end{align*}
additions.
With $\mathcal{O}(n)$ auxiliary qubits, it is possible to compute addition in $\mathcal{O}(\log n)$ layers and $\mathcal{O}(n)$ \texttt{CNOT}s \cite{takahashi2009quantum}.
Thus, the total layer complexity is $\mathcal{O}(n\log(n))$ and the total \texttt{CNOT} complexity is $\mathcal{O}(n^2)$ where $n$ corresponds to bits of precision.
Note that, if $n$ is not large, it is better to use a non-parallel version of addition to minimize overhead.

As for space complexity, we need $\mathcal{O}(n)$ bits of space.
We require an auxiliary register of size $n$ for \texttt{Mult}, and optionally another register to speed up addition.
We need registers of $n$ bits for \texttt{x}, \texttt{y} and \texttt{t}, and optionally \texttt{ang}.
Finall, we need $n-1$ bits for register \texttt{d}.
Thus, in total, we require a minimum of $5n-1$ qubits.

\begin{figure*}
     \centering
     \includegraphics[width=.95\textwidth]{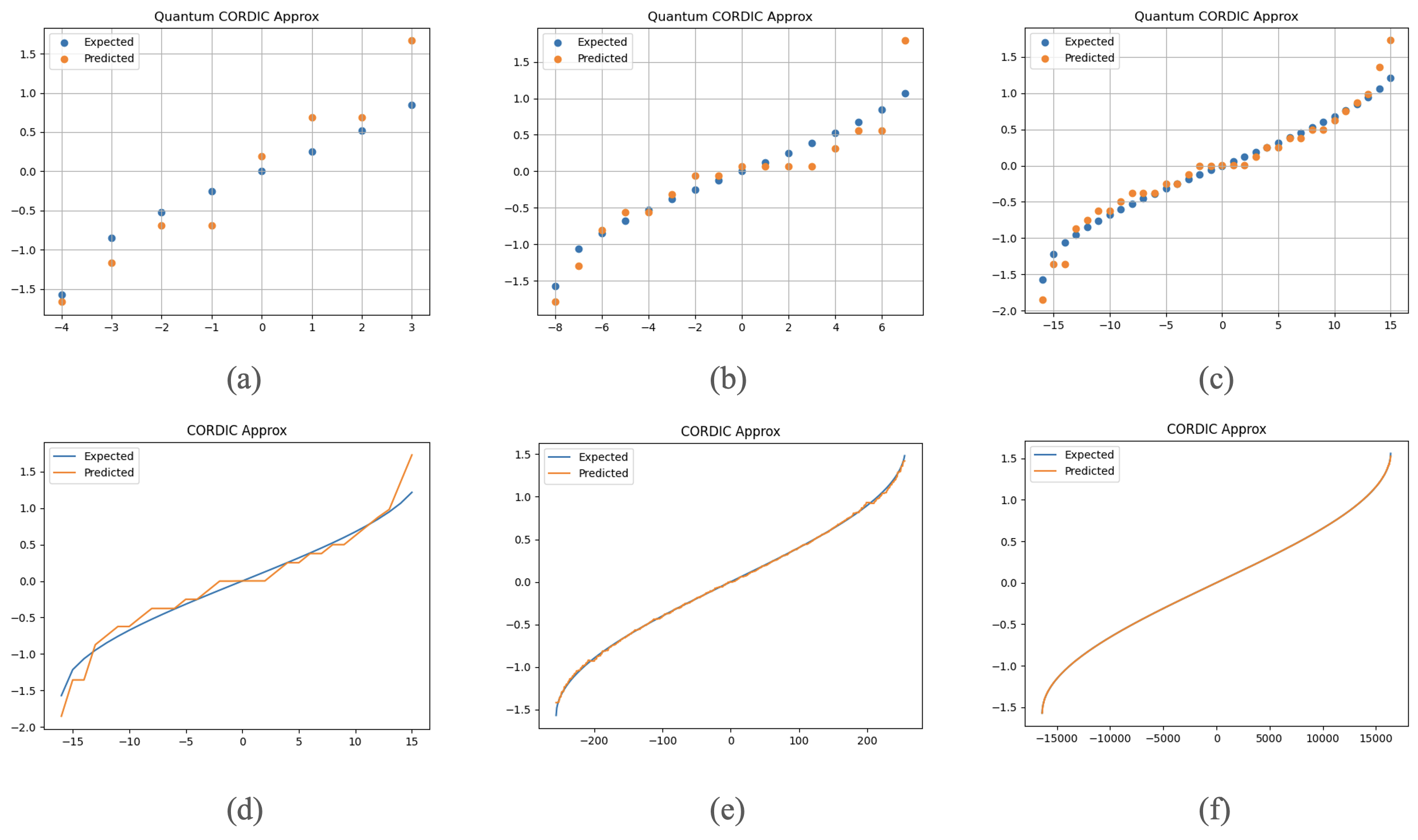}
    \caption{Results of Classical and Quantum simulations of Quantum compatible Algorithm for CORDIC Arcsine (see our GitHub repository at \href{https://github.com/iain-burge/QuantumCORDIC}{https://github.com/iain-burge/QuantumCORDIC} for further details).
    (a) Quantum simulation; $n=4$; mean error $=3.27\times10^{-1}$; max error $=8.18\times10^{-1}$. 
    (b) Quantum simulation; $n=5$; mean error $=1.83\times10^{-1}$; max error $=7.25\times10^{-1}$. 
    (c) Quantum simulation; $n=6$; mean error $=9.91\times10^{-2}$; max error $=5.13\times10^{-1}$. 
    (d) Classical simulation; $n=6$; mean error $=9.91\times10^{-2}$; max error $=5.13\times10^{-1}$. 
    (e) Classical simulation; $n=10$; mean error $=1.04\times10^{-2}$; max error $=1.51\times10^{-1}$. 
    (f) Classical simulation; $n=16$; mean error $=6.78\times10^{-4}$; max error $=3.89\times10^{-2}$.
    }
    \label{fig:Results}
\end{figure*}

In simulations, it was found that the max error for any given input can be made arbitrarily small by making $n$ larger (Figure~\ref{fig:Results}).
The full codebase can be found on our GitHub companion repository~\cite{Burge_2024}, containing a complete implementation using Qiskit as well as an identical classical implementation. 

\section{Conclusion}
\label{sec:conc}

\noindent We have introduced a Quantum algorithm for calculating \texttt{arcsin} with low space and time complexity. 
It also has applications in HHL (Harrow–Hassidim–Lloyd), approximating Shapley values, and quantum Monte Carlo methods. 
As CORDIC is most appropriate for lower precision applications, future work could apply more budget conscious implementations of CORDIC \texttt{arcsin} at the expense of asymptotic behaviour.
Beyond quantum digital-to-analog conversion, CORDIC has substantial value to bring to quantum algorithms.
In particular, developing a CORDIC module could allow for various elementary arithmetical operations to be built directly into the circuitry. 
As a result, quantum compatible CORDIC methods may find use for low to medium precision operations on early fault tolerant quantum hardware.

\bibliographystyle{plain}
\bibliography{biblio}

\appendix\label{appendix:mult}
\section*{A. \texttt{Mult} Function}

\noindent To understand the \texttt{Mult} algorithm, it is more instructive to considered its inverse, \texttt{Div} (Algorithm~\ref{al:div}). 
To simplify the following proofs, we assume infinite bits of precision in our number representations.

\begin{lemma} \label{lem:divValues}
Suppose $i$ is even, at the beginning of the $i$th iteration of Algorithm~\ref{al:div}'s \text{for} loop, the state is,
    \begin{equation*}
        \texttt{in} = z\sum_{k=0}^{F[i+1]-1} r^k,\quad 
        \texttt{aux} = z\sum_{k=0}^{F[i]-1} r^k.
    \end{equation*}
where $z$ is the initial value of \texttt{in}, and $r=-2^{-m}$.
\end{lemma}
\begin{proof}
    We proceed by induction.
    Base case: for $i=0$, at the beginning of the loop $\texttt{in}=\texttt{aux}=zr^0$.
    Thus, the base case holds.

Inductive step: suppose that for an even \hbox{$i< \left\lceil \sqrt{5}\varphi^{-m}n\right\rceil$}, the registers hold the values,
    \begin{equation*}
        \texttt{in} = z\sum_{k=0}^{F[i+1]-1} r^k,\quad 
        \texttt{aux} = z\sum_{k=0}^{F[i]-1} r^k.
    \end{equation*}
    $i$ is even so the logical statement in Line~6 is true.
    Hence, \texttt{aux} is updated to,
    \begin{equation*}
        \texttt{aux} \gets z\sum_{k=0}^{F[i]-1} r^k + (-1)^{\texttt{F}[i]} \cdot 2^{m \texttt{F}[i]} z\sum_{k=0}^{F[i+1]-1} r^k.
    \end{equation*}
    By definition~$(-1)^{\text{\texttt{F}[i]}}\cdot 2^{m \texttt{F}[i]} = r^{\text{\texttt{F}}[i]}$, thus, the new value for \texttt{aux} is
    \hbox{$z\sum_{k=0}^{F[i]-1} r^k + z\sum_{k=F[i]}^{F[i+1]+F[i]-1} r^k$}.
    By definition of the Fibonacci sequence, \hbox{$F[i+2]=F[i+1]+F[i]$}, combining the summations yields,
    \begin{equation*}
        \texttt{aux} = z\sum_{k=0}^{F[i+2]-1} r^k.
    \end{equation*}

    On the next loop, the index register holding $i$ incremented by $1$ and is now odd, so the if statement on Line~6 gets false.
    Hence, the next operation updates \texttt{in},
    \begin{equation*}
        \texttt{in} \gets z\sum_{k=0}^{F[i+1]-1} r^k + (-1)^{\texttt{F}[i+1]} \cdot 2^{m \texttt{F}[i+1]} z\sum_{k=0}^{F[i+2]-1} r^k
    \end{equation*}
    Similarly to the previous iteration, we find \texttt{in} is equal to
    \hbox{$z\sum_{k=0}^{F[i+1]-1} r^k + z\sum_{k=F[i+1]}^{F[i+2]+F[i+1]-1} r^k$}.
    Bringing the summations together shows,
    \begin{equation*}
        \texttt{in} = z\sum_{k=0}^{F[i+3]-1} r^k .
    \end{equation*}
    The index register is incremented once again, such that it holds $i+2$.
    By induction, the result holds.
\end{proof}

\begin{remark}
    Recall, for $r\in(0,1)$, and an integer $J$, the geometric sum formula gives,
    \begin{equation*}
        \sum\limits_{k=0}^{J-1} r^k = \frac{1-r^J}{1-r} = \frac{1-2^{-Jm}}{1+2^{-m}}.
    \end{equation*}
\end{remark}

\begin{algorithm}
  \caption{Div($\texttt{in},\texttt{aux},m$): $\texttt{in} \gets (1+2^{-m})^{-1}\texttt{in}$}\label{al:div}
  \begin{algorithmic}[1]
    \Procedure{Div}{$\texttt{in}$, \texttt{aux}, $m$}
      \Comment{$\texttt{in}$ is a register of size $n$, \texttt{aux} is a auxiliary register of size $n$ with near 0 value, $m$ is the right shift}
      \State{\texttt{F} $\gets [1,1,2,3,5,8,13,\dots]$} \Comment{Fibbonacci Sequence}
      \State{$\texttt{\#iter} \gets 2\left\lceil\log_\varphi(\sqrt{5} n/m)/2\right\rceil$} \Comment{$\varphi = (1+\sqrt{5})/2$}
      \State{$\texttt{aux} \gets \texttt{aux}+\texttt{in}$} \label{line:initStart}
      \For{$i=0$, $i < \texttt{\#iter}$, $\texttt{++}i$} \label{line:buildX}
        \If{$i$ even} \label{line:ifIeven}
          \State $\texttt{aux} \gets \texttt{aux}+(-1)^{\texttt{F}[i]}(\texttt{in}\gg(m\cdot \texttt{F}[i]))$
        \Else
          \State $\texttt{in} \gets \texttt{in}+(-1)^{\texttt{F}[i]}(\texttt{aux}\gg(m\cdot \texttt{F}[i]))$ 
        \EndIf
      \EndFor
      \State $\texttt{aux} \gets \texttt{aux}-\texttt{in}$\label{line:div-EmptyAux}
    \EndProcedure
  \end{algorithmic}
\end{algorithm}

\begin{theorem}
    Assuming $J$ iterations, the final state is,
    \begin{equation*}
        \texttt{in} \approx\frac{z}{1-2^{-m}},\quad
        \texttt{aux} \approx 0,
    \end{equation*}
    with an accuracy of $\mathcal{O}\left(m\varphi^J\right)$ bits.
\end{theorem}
\begin{proof}
    First note the state prior to Line~10 is,
    \begin{equation*}
        \texttt{in} = z\sum_{k=0}^{\texttt{F}[J+1]-1} r^k,\quad 
        \texttt{aux} = z\sum_{k=0}^{\texttt{F}[J]-1} r^k.
    \end{equation*}
    For simplicity, we assume $J$ is even  and greater than or equal to two.
    Then, Line~10 results in state,
    \begin{equation*}
        \texttt{aux} \gets z\sum_{k=0}^{\texttt{F}[J]-1} r^k - z\sum_{k=0}^{\texttt{F}[J+1]-1} r^k.
    \end{equation*}
    As a result, \texttt{aux} is equal to
    \hbox{$- z\sum_{k=\texttt{F}[J]}^{\texttt{F}[J+1]-\texttt{F}[J]-1} r^k$}.
    This can be simplified to,
    \hbox{$-z r^{\texttt{F}[J]} \sum_{k=0}^{\texttt{F}[J-1]-1} r^k$}.
    Then, the geometric sum formula gives,
    \begin{equation*}
        \texttt{aux} = -z r^{\texttt{F}[J]} \frac{1-r^{\texttt{F}[J-1]}}{1-r}.
    \end{equation*}
    Note that $\varphi^{k-1} < \texttt{F}[k]$ for positive $k$.
    Thus,
    \begin{equation*}
        \texttt{aux} < -z 2^{-\varphi^{J-1}m} \frac{1-2^{-\varphi^{J-2} m }}{1-2^{-m}}.
    \end{equation*}
    Note that $z$ is always less than $2$ as argued at the beginning of Section~\ref{sec:qCORDIC}.
    Since $m\geq 1$, $J\geq 2$, we have \hbox{$|aux| \leq 2^{-\varphi^{J-1}m + 2}$}.
    Thus, the error introduced to \texttt{aux} in terms of bit precision is exponentially small.
    
\noindent  On the other hand, the \texttt{in} register state is,
    \begin{equation*}
        \texttt{in} = z \frac{1-2^{-mF[J]}}{1+2^{-m}}. 
    \end{equation*}
    Thus, we have an absolute error of less than,
    \hbox{$2\cdot 2^{-mF[J]} \leq 2^{-m\varphi^{J-1}}$}.
    Therefore, the absolute error is exponentially small with respect to number of iterations.
\end{proof}

\balance

\end{document}